\documentclass{article}
\usepackage{dcolumn}
\usepackage{bm}
\usepackage{verbatim}       

\usepackage[dvips]{graphicx}
\usepackage{amsthm}
\usepackage{amssymb}
\usepackage{bm}
\usepackage{amstext}
\usepackage{psfrag}
\usepackage{amsmath}
\usepackage{amsfonts}

\newfont{\bb}{msbm10 at 12pt}

\newcommand{\bd}{\begin{definition}}                
\newcommand{\ed}{\end{definition}}                  
\newcommand{\bc}{\begin{corollary}}                 
\newcommand{\ec}{\end{corollary}}                   
\newcommand{\bl}{\begin{lemma}}                     
\newcommand{\el}{\end{lemma}}                       
\newcommand{\bp}{\begin{proposition}}            
\newcommand{\ep}{\end{proposition}}                
\newcommand{\bere}{\begin{remark}}                  
\newcommand{\ere}{\end{remark}}                     

\newcommand{\bt}{\begin{theorem}}
\newcommand{\et}{\end{theorem}}

\newcommand{\be}{\begin{equation}}
\newcommand{\ee}{\end{equation}}

\newcommand{\bit}{\begin{itemize}}
\newcommand{\eit}{\end{itemize}}
\newtheorem{theorem}{Theorem}[section]
\newtheorem{corollary}[theorem]{Corollary}

\newtheorem{lemma}[theorem]{Lemma}
\newtheorem{proposition}[theorem]{Proposition}
\theoremstyle{definition}
\newtheorem{definition}[theorem]{Definition}
\theoremstyle{remark}
\newtheorem{remark}[theorem]{Remark}

\begin{document}
%

\title{In a distinguishing spacetime the horismos relation generates the causal relation}

\author{E. Minguzzi \footnote{Dipartimento di Matematica Applicata, Universit\`a degli Studi di Firenze,  Via
S. Marta 3,  I-50139 Firenze, Italy. E-mail:
ettore.minguzzi@unifi.it}}

\date{}
 \maketitle

\begin{abstract}
\noindent It is proved that in a distinguishing spacetime the
horismos relation {$E^{+}\!\!=\!J^{+}\backslash I^{+}$} generates
the causal relation $J^{+}$. In other words two causally related
events are joined by a chain of horismotically related events, or
again, the causal relation is the smallest transitive relation
containing the horismos relation. The result is sharp in the sense
that distinction can not be weakened to future or past distinction.
Finally, it is proved that a  spacetime in which the horismos
relation generates the causal relation is necessarily non-total
imprisoning.
\end{abstract}



\section{Introduction}
In a spacetime $(M,g)$ (a $C^{r}$ connected, time-oriented
Lorentzian manifold, $r\in \{3, \dots, \infty\}$ of arbitrary
dimension $n\geq 2$ and signature $(-,+,\dots,+)$) we write as usual
$p\ll q$ if there is a future directed timelike curve joining $p$ to
$q$; $p<q$ if there is a future directed causal curve joining $p$ to
$q$; $p\le q$ if $p<q$ or $p=q$ and finally\footnote{In the
definition of $\to$ or $E^{+}$ one can either choose reflexivity,
$p\to p$ for every $p$, as in \cite{kronheimer67}, or to make the
equality $E^{+}=J^{+}\backslash I^{+}$ always true as in
\cite{garciaparrado05,minguzzi06c}. Here we choose the latter
possibility, this being  equivalent to Kronheimer and Penrose's
definition for chronological spacetimes.} $p \to q$ if $p\le q$ but
$p\not\!\ll q$. A well known theorem \cite{hawking73,beem96}
establishes that $p\to q$ iff either $p=q$ and chronology holds at
$p$, or $p$ and $q$ are connected by an achronal lightlike geodesic
segment (thus without conjugate points in its interior). The
chronological, causal, and horismos relations on $M$ are defined by
\begin{align}
I^{+}&=\{(p,q): p\ll q\},\\
J^{+}&=\{(p,q): p\le q\}, \\
E^{+}&=\{(p,q): p \to q\}=J^{+}\backslash I^{+},
\end{align}
and of course they are subsets of $M\times M$. While $I^{+}$ and
$J^{+}$ are transitive, the relation $E^{+}$ in general is not as
the composition of two lightlike achronal geodesics segments may
connect chronologically related events.

The smallest reflexive transitive relation which contains $E^{+}$
is\footnote{In the notation of \cite{minguzzi06c} $x \le^{(\to)} y$
iff $(x,y)\in T^{+}$.}
\begin{equation} \label{nod}
T^{+}\equiv\bigcup_{n=0}^{+\infty} (E^{+})^n ,
\end{equation}
where it is understood that $(E^{+})^0=\Delta$ is the diagonal of $M
\times M$. We have $(p,q)\in T^{+}$ iff $p=q$ or there is a chain of
horismotically related events which connects $p$ to $q$. Note that
the smallest transitive relation containing $E^{+}$ is
$\bigcup_{n=1}^{+\infty} (E^{+})^n$. In a chronological spacetime
$E^{+}$ is itself reflexive, that is $\Delta \subset E^{+}$, so that
the union in Eq. (\ref{nod}) can start from $n=1$, and hence $T^{+}$
becomes also the smallest transitive relation containing $E^{+}$.
Moreover, in this case $(E^{+})^n\subset (E^{+})^{n+1}$ so that the
transitive relation $T^{+}$ can also be denoted $E^{+\infty}$
coherently with the notation of \cite{minguzzi07b}.

Since piecewise lightlike geodesics are causal curves we have
$T^{+}\subset J^{+}$ but the converse inclusion may not hold. In the
equality case one can also easily recover the chronological relation
as $I^{+}=T^{+}\backslash E^{+}$. This approach in which the
relations $I^{+}$, $J^{+}$ and $E^{+}$ are recovered from just one
of them, in this case $E^{+}$, can be useful to put causality theory
in an order theoretic framework in which one focuses on just one
relation (usually a partial order).

Under strong causality it is known that $T^{+}= J^{+}$. This result
was stated by Kronheimer and Penrose \cite[Sect. 2.1]{kronheimer67}
(they state that this equivalence holds if the Alexandrov topology
is Hausdorff which is equivalent to strong causality) and a detailed
proof can be found for instance in \cite[Def. 2.22 and Th.
3.24]{minguzzi06c}. This paper shows that distinction suffices to
guarantee the equality $T^{+}= J^{+}$, and in fact that this result
is sharp in the sense that distinction can not be weakened to future
or past distinction.
%

Recall that a spacetime is distinguishing if ``$I^{+}(x)=I^{+}(y)$
or $I^{-}(x)=I^{-}(y) \Rightarrow x=y$''. An open neighborhood $U$
distinguishes $p\in U$ if every causal curve $\gamma:I\to M$,
passing through $p$ intersects $U$ only once (i.e. in a connected
subset of its domain $I$). A spacetime is distinguishing iff every
point admits  arbitrarily small distinguishing open
neighborhoods\footnote{The proof of this statement contained in
\cite[Lemma 3.10]{minguzzi06c} is given in the future distinguishing
case, the case considered here being analogous.}
\cite{hawking73,minguzzi06c}. Analogous definitions in the past and
future case exist \cite{minguzzi06c}.

The following result holds true

\begin{theorem}
In a past or future distinguishing spacetime the conformal structure
is determined by the horismos relation $E^{+}$. In other words if
$(M,g_1)$ and $(M,g_2)$ have the same horismos relation, i.e.
$E^{+}_1=E^{+}_2$, then the metrics are conformally related, namely
$g_2=\Omega g_1$, with $\Omega: M \to (0,+\infty)$.
\end{theorem}

\begin{proof}
Future case. In a future distinguishing spacetime every point admits
arbitrarily small future distinguishing open neighborhoods. Given
$p\in M$ we can find a neighborhood $U$ which distinguishes $p$ in
the future, contained in a convex neighborhood. By future
distinction at $p$, we have $J^{+}_U(p)=J^{+}(p)\cap U$ and
$I^{+}_U(p)=I^{+}(p)\cap U$, thus $E^{+}_U(p)=E^{+}(p)\cap U$. But
$E^{+}_U(p)$ is the image of exponential map based at $p$, in the
spacetime $(U,g\vert_U)$, of the future light cone at $p$. Given
$E^{+}(p)$ one determines the light cone at $p$ through the inverse
 of the exponential map, and thus  $E^{+}$ fixes the conformal structure of the spacetime (see \cite[Appendix D]{wald84}).
\end{proof}

This result is  analogous to Malament's theorem \cite{malament77b}
which states that in a past or future distinguishing spacetime the
causal relation determines the conformal structure (see also
\cite[Prop 3.13]{minguzzi06c}).
%

Since in a past or future  distinguishing spacetime $E^{+}$ fixes
the conformal structure, it also clearly determines  the causal
relation. The non-trivial result established by this work is that in
a distinguishing spacetime the causal relation is actually given by
$T^{+}$ namely by the smallest reflexive transitive relation
containing $E^{+}$. Figure 37 of \cite{hawking73} shows an example
of past distinguishing spacetime for which $T^{+}\ne J^{+}$, despite
the fact that in this spacetime $E^{+}$ fixes the conformal
structure (the point $p$ in the figure is such that $T^{+}(p)$
equals $E^{+}(p)$ and is given by a portion of the middle lightlike
geodesic, thus every point above $p$ although belonging to
$I^{+}(p)$ is not in $T^{+}(p)$). This example shows that our result
is sharp: distinction can not be weakened to past or future
distinction. Nevertheless, distinction is not equivalent to the
equality $T^{+}=J^{+}$ as the discussion at the end of the next
section shows.


\section{The proof and a counterexample to the other direction}
We start with a couple of lemmas.

\begin{lemma}
Assume that $(M,g)$ is distinguishing at $p$, then $p$ admits
arbitrarily small globally hyperbolic neighborhoods that distinguish
$p$.
\end{lemma}

\begin{proof}
Let $W \ni p$ be an arbitrary small strongly causal neighborhood.
There is a neighborhood $V \subset W$ that distinguishes $p$. Let
$U\ni p$, $U\subset V$, be a neighborhood which is both causally
convex  with respect to $W$ and globally hyperbolic (see
\cite{minguzzi06c} remark 2.15). Let $\gamma: I \to M$ be a causal
curve passing through $p$ and let $r \in \gamma\cap U$. It is either
$p\le r$ or $r \le p$. Let us consider the former case, the latter
being analogous. The curve $\gamma$ between $p$ and $r$ is
necessarily contained in $V$, as $p,r \in V$, and $V$ distinguishes
$p$. Thus it is contained in $W$ and since $U$ is casually convex in
$W$, and $p,r \in U$, it is contained in $U$. As $r$ is arbitrary
every causal curve passing through $p$ intersects $U$ only once,
that is, $U$ distinguishes $p$.
%
%
%
\end{proof}

The next lemma is the crucial one in the proof.

\begin{lemma}
Let $V\ni p$ be a globally hyperbolic neighborhood which
distinguishes $p$ contained in a convex neighborhood. Then for every
$q \in I^{+}_V(p)$, the set $E^{+}_V(p)\cap E^{-}_V(q)$ is non-empty
and for every $r \in E^{+}_V(p)\cap E^{-}_V(q)$ we have $p \to r \to
q$ where $\to$ stands for the causal relation $E^{+}$ in $(M,g)$.
Analogously, for every $q \in I^{-}_V(p)$, $E^{+}_V(q)\cap
E^{-}_V(p)$ is non-empty and if $r \in E^{+}_V(q)\cap E^{-}_V(p)$ we
have $q \to r \to p$.
\end{lemma}

\begin{proof} Let us prove the former case, the latter being
analogous. First $E^{+}_V(p)\cap E^{-}_V(q)$ is a closed subset of
the compact $J^{+}_V(p)\cap J^{-}_V(q)$ thus it is compact. Let us
denote this set with $S$. Since $V$ future distinguishes $p$,
$J^{+}_V(p)=J^{+}(p)\cap V$ and $I^{+}_V(p)=I^{+}(p)\cap V$  thus
$E^{+}_V(p)=E^{+}(p)\cap V$. We conclude that for every $r \in S$ it
is $p \to r$. Note that a lightlike geodesic generator of
$E^{+}_V(p)$ extended towards the future cannot enter $I^{+}(p)$
before escaping $V$, i.e. it cannot enter $I^{+}_V(p)$, as $V$ is
contained in a convex neighborhood. As a consequence, as long as it
stays in $V$ it belongs to $E^{+}_V(p)$. In fact it must reach the
boundary of the compact set $J^{+}_V(p)\cap J^{-}_V(q)$, and thus
$E^{-}_V(q)$, otherwise it would be totally future imprisoned in a
compact which is impossible because $V$ is globally hyperbolic. Thus
the geodesic generators of $E^{+}_V(p)$, once extended to the future
intersect $S$, and conversely every point of $S$ is connected to $p$
by a lightlike geodesic. In particular $S$ is non-empty.

Let  $r\in S$ and let us prove that  $r \to q$. Assume not, then $r
\in S$ can be connected to $q$ with a timelike curve $\gamma$
necessarily not entirely contained in $V$. The causal curve joining
$p$ to $r$ and $r$ to $q$ along $\gamma$ is not entirely contained
in $V$ and contradicts the fact that $V$ distinguishes $p$.

\end{proof}

%

\begin{theorem} \label{jhg}
In a distinguishing spacetime,  $J^{+}=T^{+}(=E^{+\infty})$.
\end{theorem}

\begin{proof}
$\Rightarrow$. The direction $T^{+} \subset J^{+}$  is obvious. Let
us prove $J^{+}\subset T^{+}$. Let $x \le y$, if $y \notin I^{+}(x)$
then $x \to y$ and there is nothing to prove. Thus assume $x \ll y$,
and let $\gamma$ be a timelike curve connecting $x$ to $y$. For
every $z \in \gamma$ let $V_z$ be a globally hyperbolic neighborhood
that distinguishes $z$ contained in a convex neighborhood
$W_z\supset V_z$. Extract a finite subcovering $\{V_{z_n}\}$ from
the covering $\{V_{z}, z \in \gamma\}$
so that $V_{z_{i-1}}\cap V_{z_i}\ne \emptyset$. Thus taking $q_i \in
V_{z_{i-1}}\cap V_{z_i}$, $z_{i-1}$ can be joined to $q_i$ with a
causal curve made of two achronal lightlike segments, and $q_i$ can
be joined to $z_i$ with a causal curve made of two achronal
lightlike segments. Joining all the pieces the searched curve is
obtained. In conclusion $(x,y) \in T^{+}$.

\end{proof}

It is natural to ask if the equality $T^{+}=J^{+}$ implies
distinction. The answer is negative. A counterexample is given by
the spacetime of figure 2 in \cite{minguzzi07e}, although it is not
easy to grasp why this spacetime does indeed provide a
counterexample. The next theorem establishes some features that any
counterexample should have. Recall that a spacetime is non-total
imprisoning if no inextendible causal curve is contained in a
compact \cite{beem76,minguzzi07f}. Distinction implies non-total
imprisonment which implies causality. Recall also that a lightlike
line is an inextendible achronal causal curve and thus a lightlike
geodesic.

\begin{theorem}
If a spacetime $(M,g)$ is such that  $T^{+}=J^{+}$ then it is
non-total imprisoning. Moreover, either (a) $(M,g)$ is future
distinguishing or (b) there is a lightlike line $\sigma$ and a last
point (in the future direction) $w\in \sigma$ with the property that
$\sigma \subset \overline{I^{+}(w)}$. A past version also holds.
\end{theorem}

\begin{proof}
 Let us assume that $T^{+}=J^{+}$. The spacetime is chronological
otherwise there would be a point $p$ inside the chronology violating
set. Since the chronology violating  class $I^{+}(p)\cap I^{-}(p)$
is open and contains $p$ no achronal lightlike geodesic segment can
start from $p$, thus $J^{+}(p)=T^{+}(p)=\{p\}$, which is a
contradiction.

The spacetime is non-total imprisoning, indeed otherwise there would
be a totally imprisoned causal curve, and because of  chronology
there would be an achronal (minimal invariant) set $\Omega$
generated by lighlike lines such that all the points in the set have
the same chronological past and future \cite[Theorem
3.9]{minguzzi07f}. Let $x\in \Omega$, since through $x$ there passes
a lightlike line $\gamma$ included in $\Omega$, we can take $y \in
\gamma$ before $x$. Any point $z \in J^{+}(x)$ which does not belong
to $\gamma$ belongs to $I^{+}(y)$, because if $x$ is connected to
$z$ by a lightlike geodesic segment then it does not join smoothly
with $\gamma$ at $x$ as $z \notin \gamma$. But $I^{+}(y)=I^{+}(x)$
thus $E^{+}(x)\subset \gamma$. Since the same argument holds for
every point in $\gamma$, $T^{+}(x) \subset \gamma$, but $I^{+}(x)$
contains points which are not in $\gamma$ because $\Omega$ is
achronal, a contradiction.

%

Let us assume that the spacetime is not future distinguishing, then
we can find a point $x\in M$ and a sequence of causal curve
$\sigma_n$ of starting point $x$ and ending point $z_n$ with $z_n\to
x$, and a neighborhood $V\ni x$ such that none of the $\sigma_n$ is
contained in $V$. By causality an application of the limit curve
theorem gives us a future inextendible limit continuous causal curve
$\sigma^x$ starting from $x$ and a past inextendible limit
continuous causal curve $\sigma^z$ ending at $x$. The causal curve
$\sigma=\sigma^x\circ \sigma^z$ is inextendible and achronal, thus a
lightlike line, otherwise chronology would be violated. Note that
since all the points of $\sigma$ are limit points of a subsequence
of $\sigma_n$, $\sigma\subset \overline{I^{+}(x)}$. Note that if a
point $y\in \sigma$ has the property $\sigma\subset
\overline{I^{+}(y)}$ then the same is true for all the points which
come before $y$ on $\sigma$. Moreover, if a point $y$ has this
property then it has the same chronological future of the points
which come before it on $\sigma$. As a consequence for a point $y$
sharing this property $E^{+}(y)\subset \sigma$, indeed a point $z\in
E^{+}(y)\backslash \sigma$ would be chronologically related with the
points before $y$ (the lightlike geodesics do not join smoothly at
$y$) and hence with $y$, a contradiction. Thus not all the points on
$\sigma$ can share the mentioned property otherwise $T^{+}(y)\subset
\sigma$ while $I^{+}(y)$, $I^{+}(y)\subset T^{+}(y)$, has clearly
points which do not belong to $\sigma$. Thus there is a point $w \in
\sigma$ such that all the points $y\in \sigma$ before it share the
mentioned property while those after it do not. We have to show that
$w$ shares the property. Since for every  $y\in \sigma$ before $w$,
we have $E^{+}(y)\subset \gamma$ we also have $T^{+}(y) \subset
J^{+}(w)$ because the chain of $E^{+}$ related events starting from
$y$ can leave $\sigma$ only after having left the portion of
$\sigma$ before $w$. Thus $I^{+}(y)\subset T^{+}(y)\subset J^{+}(w)$
and hence $y \in \overline{J^{+}(w)}$. As it holds for every $y$
before $w$ and it clearly holds for every points after $w$ we have
$\sigma \subset \overline{I^{+}(w)}$.

\end{proof}

\section{Conclusions}
It has been proved that in a distinguishing spacetime the horismos
relation generates the causal relation, namely $E^{+\infty}=J^{+}$,
and that the result is optimal in the sense that distinction can not
be weakened to future or past distinction.

As previously mentioned this result may be useful in all those
 approaches which try to replace the spacetime  with a
partially ordered set.  This is the route taken for instance in the
Causal Set Theory of quantum gravity
\cite{bombelli87,sorkin91,henson06}. Here an inspiring result is the
already mentioned Malament's theorem \cite{malament77b} which states
that in a distinguishing spacetime the causal relation determines
the conformal structure of spacetime. The conformal factor can be
fixed through the notion of volume on spacetime. If one picks up
points from the spacetime in a random way so that equal volumes have
equal average number of picked up points, and connects them through
the causal relation one obtains a causal set, that is, a discrete
version of the original spacetime. The idea is that one can start
from an abstract notion of causal set, regard it as fundamental, and
study in which limit it approximates the usual spacetime continuum.

If the distinction property suitably extended to the discrete domain
holds then we can have the reasonable hope to get the right
continuum limit as the abstract causal relation imbedded in the
causal set would pass to a continuum causal relation which in turn
determines the light cone structure. It is interesting to note that
at least the property of non-total imprisonment is always satisfied
in a causal set, indeed the local finiteness property
\cite{bombelli87,sorkin91,henson06} implies that every causal curve
entering a finite volume would pass through the finite number of
vertices there contained to finally escape.


The result proved in this work may be useful to build up an abstract
framework in which the horismos rather that the causal relation,
plays the fundamental role. This possibility is already suggested by
causal set theory in which there seems to be a somewhat unexplored
analogy (but see \cite{philpott08}) between the concept of horismos
relation on the continuum side and that of {\em link}
\cite{johnston08} on the causal set (discrete) side. Both of them,
in the respective contexts, generate the causal relation.

\section*{Acknowledgments}
This work has been partially supported by GNFM of INDAM. My interest
on a possible improvement of Kronheimer and Penrose's result arose
while writing the review \cite{minguzzi06c} jointly with M.
S\'anchez. Useful conversations with him on this problem are
acknowledged.


\end{document}